\documentclass[conference,letterpaper]{IEEEtran}

\usepackage{cite}
\usepackage{amsmath,amssymb,amsfonts}
\usepackage{algorithmic}
\usepackage{graphicx}
\usepackage{textcomp}
\usepackage{xcolor}
\usepackage{amsthm}

\def\BibTeX{{\rm B\kern-.05em{\sc i\kern-.025em b}\kern-.08em
    T\kern-.1667em\lower.7ex\hbox{E}\kern-.125emX}}
\begin{document}

\title{Concave Aspects of Submodular Functions}

\author{%
  \IEEEauthorblockN{Rishabh Iyer}
  \IEEEauthorblockA{University of Texas at Dallas, CSE Department\\
                    2601 N. Floyd Rd. MS EC31\\
                    Richardson, TX 75083 \\
                    Email: rishabh.iyer@utdallas.edu}
  \and
  \IEEEauthorblockN{Jeff Bilmes}
  \IEEEauthorblockA{University of Washington, ECE Department\\
                    185 E Stevens Way NE\\
                    Seattle, WA 98195 \\
                    Email: bilmes@ee.washington.edu}
}

\newcommand{\union}{\cup}
\newcommand{\inter}{\cap}
\newcommand{\symmdiff}{\bigtriangleup} 
\newcommand{\lovasz}{Lov\'asz}
\newcommand{\Ph}{\widehat{\mathcal{P}}}
\newcommand{\lex}[1]{{\ensuremath \breve #1}}
\newcommand{\mex}[1]{{\ensuremath \tilde #1}}
\newcommand{\cex}[1]{{\ensuremath \invbreve #1}}

\newcommand{\ggrow}{\hat{g}} 
\newcommand{\gshrink}{\check{g}}  
\newcommand{\gthree}{\bar{g}} 

\newtheorem{theorem}{Theorem}
\newtheorem{lemma}[theorem]{Lemma}
\newtheorem{proposition}[theorem]{Proposition}
\newtheorem{observation}{Observation}[section]
\newtheorem{corollary}[theorem]{Corollary}
\newtheorem{definition}{Definition}[section]
\newtheorem{example}{Example}[section]

\maketitle

\begin{abstract}
Submodular Functions are a special class of set functions, which generalize several information-theoretic quantities such as entropy and mutual information~\cite{bilmeskarbasiISIT2018}. Submodular functions have subgradients and subdifferentials~\cite{fujishige1984subdifferential} and admit polynomial time algorithms for
  minimization, both of which are fundamental
  characteristics of convex functions. Submodular functions also show
  signs similar to concavity. Submodular function maximization, though
  NP hard, admits constant factor approximation guarantees and concave
  functions composed with modular functions are submodular. In this
  paper, we try to provide a more complete picture on the relationship
  between submodularity with concavity.  We  characterize the
  superdifferentials and polyhedra associated with upper bounds and provide optimality conditions for submodular maximization using the superdifferentials. This paper is a concise and shorter version of our longer preprint~\cite{iyer2015polyhedral}.
\end{abstract}


\section{Introduction}
Long known to be an important property for problems in
combinatorial optimization, economics, operations research, and game
theory, submodularity is gaining popularity in a number of new areas including Machine Learning and Information Theory. Submodular Functions are defined over subsets $X$ of a ground-set $V$. Let $V = \{1, \cdots, n\}$ be a set of items, then a set function $f: 2^V \to \mathbb R$ over a ground set $V = \{1, 2,
\cdots, n\}$ is \emph{submodular} if for all subsets $S, T \subseteq
V$, it holds that, $f(S) + f(T) \geq f(S \cup T) + f(S \cap
T)$. Equivalently, a submodular set function satisfies
\emph{diminishing marginal returns}: Let $f(j | S) = f(S \cup \{ j \}) -
f(S)$ denote the marginal cost of element $j\in V$ with respect to $S
\subseteq V$.\footnote{We also use this notation for sets $A,B$ as in $f(A|B) = f(A\cup B) - f(B)$.}
The diminishing returns property states that, $f(j | S) \geq f(j | T), \forall S \subseteq T \text{ and } j \notin T$. Given a set of random variables $\mathcal X_1, \cdots, \mathcal X_n$, define $H(\mathcal X_S)$ as the joint Entropy of variables indexed by subset $S$: i.e. $\mathcal X_S = \{\mathcal X_i, i \in S\}$. Its easy to see that $H(\mathcal X_S)$ is submodular~\cite{bilmeskarbasiISIT2018}. Similarly, the Mutual Information $I(\mathcal X_S; \mathcal X_{V \backslash S}) $ is also a submodular function. Submodular Functions are increasingly becoming prevalent in machine learning applications including data selection~\cite{wei2015submodularity,liu2015svitchboard}, summarization~\cite{gygli2015video,tsciatchek14image,bairi2015summarization}, observation selection and sensor placement~\cite{krause2008efficient} to name a few.

Submodular functions have been strongly associated with convex functions, to the extent that submodularity is sometimes regarded as a discrete analogue
of convexity~\cite{fujishige2005submodular}. This relationship is evident by the fact that submodular function minimization is polynomial time (akin to convex minimization). A number of recent results, however, make this relationship much more formal. For example, similar to convex functions, submodular functions have tight modular lower bounds and admit a sub-differential characterization~\cite{fujishige1984subdifferential}. Moreover, it is possible~\cite{fujishige1984theory} to provide optimality conditions for submodular minimization, in a manner analogous to the Karush-Kuhn-Tucker (KKT) conditions from convex programming. In addition, submodular functions also admit a natural convex extension, known as the \lovasz{} extension, which is easy to evaluate~\cite{lovasz1983} and optimize. \looseness-1

Submodular functions also have some properties, which are unlike convexity, but perhaps more akin to concavity. Submodular function maximization is known to be NP hard. However, there exist a number of constant factor approximation algorithms based on simple greedy or local search hueristics~\cite{janvondrak, lee2009non, nemhauser1978} and some recent continuous approximation methods~\cite{chekuri2011submodular}. This is unlike convexity where 
maximization can be hopelessly difficult~\cite{sahni1974computationally}.
Furthermore, submodular functions have a diminishing returns property which is akin to concavity, and concave over modular functions are known to be submodular. In addition submodular function has been shown to have tight modular upper bounds~\cite{rkiyersemiframework2013,jegelkacvpr,rkiyersubmodBregman2012,rkiyeruai2012,nipssubcons2013}, and as we show, form superdifferentials and supergradients like concave functions. The multi-linear extension of a submodular function, which has become useful recently~\cite{chekuri2011submodular} in the context of submodular maximization, is known to be concave when restricted to a particular direction. All these seem to indicate that submodular functions are related to both convexity and concavity, and have some strange properties enabling them to get the best of both classes of functions. We formalize all these relationships in this paper.\looseness-1

\subsection{Our Contributions}
The main contributions of this work is in providing the first systematic theoretical study related to polyhedral aspects of submodular function maximization and connections to concavity. We show that submodular functions have tight modular (additive) upper bounds, thereby proving the existence of the superdifferential of a submodular function. We show that characterizing this superdiffereitial is NP hard in general. However, we provide a series of (successively tighter) outer and also inner polyhedral bounds, all obtainable in polynomial time, and also show that we can obtain some specific practically useful supergradients in polynomial time.
We then show how we can define forms of
 optimality conditions for submodular maximization through the submodular superdifferential. We also show how optimality conditions related to approximations to the superdifferential lead to a number of familiar approximation guarantees for these problems. 

\section{Submodularity and Convexity} \label{subconvex}
Most of the results in this section are from~\cite{lovasz1983, fujishige2005submodular} and the references contained therein, so for more details please refer to these texts. We use this section to review existing work on the connections between submodularity and convexity, and to help contrast these with the corresponding results on the connections between submodularity and concavity.

\subsection{Submodular (Lower) Polyhedron}
\label{submodpoly}
For a submodular function $f$, the (lower) submodular
polyhedron and the base polytope of a
submodular function \cite{fujishige2005submodular} are defined,
respectively, as: $\mathcal P_f = \{ x : x(S) \leq f(S), \forall S \subseteq V \}$, and 
$\mathcal B_f = \mathcal P_f \cap \{ x : x(V) = f(V) \}$. The submodular polyhedron has a number of interesting properties. An important property of the polyhedron is that the extreme points and facets can easily be characterized even though the polyhedron itself is described by a exponential number of inequalities. In fact, surprisingly, every extreme point of the submodular polyhedron is an extreme point of the base polytope. These extreme points admit an interesting characterization in that they can be computed via a simple greedy algorithm
--- let $\sigma$ be a permutation of $V = \{1, 2, \cdots, n\}$. Each such permutation defines a chain with elements $S^{\sigma}_0 = \emptyset$, 
$S^{\sigma}_i = \{ \sigma(1), \sigma(2), \dots, \sigma(i) \}$ such that $S^{\sigma}_0 \subseteq S^{\sigma}_1 \subseteq \cdots \subseteq S^{\sigma}_n$. This chain defines an extreme point
$h^{\sigma}$ of $\mathcal P_f$ with entries: $h^{\sigma}(\sigma(i)) = 
f(S^{\sigma}_i) - f(S^{\sigma}_{i-1})$. 
Each permutation of $V$ characterizes an extreme point of $\mathcal P_f$ and all possible extreme points of $\mathcal P_f$ can be characterized in this manner~\cite{fujishige2005submodular}. Given a submodular function $f$ such that $f(\emptyset) = 0$, the condition that $x \in \mathcal P_f$ can be checked in polynomial time for every $x$. 
\begin{proposition} \cite{fujishige2005submodular}
Given a submodular function $f$, checking if $x \in \mathcal P_f$ is equivalent to the condition $\min_{X \subseteq V} [f(X) - x(X)] \geq 0$, which can be checked in poly-time.
\end{proposition}

\subsection{The Submodular Subdifferential}
\label{sec:subm-subd}
The subdifferential $\partial_f(X)$ of a submodular set function $f:
2^V \to \mathbb{R}$ for a set $Y \subseteq V$ is defined 
\cite{fujishige1984subdifferential, fujishige2005submodular}
analogously
to the subdifferential of a continuous convex function:
$\partial_f(X) = \{x \in \mathbb{R}^n: f(Y) - x(Y) \geq f(X) - x(X)\;\text{for all } Y \subseteq V\}$. 
The polyhedra above can be defined for any (not necessarily submodular) set function. When the function is submodular however, it can be characterized efficiently.
Firstly, note that for normalized submodular functions, for any $h_X \in \partial_f(X)$, we have $f(X) - h_X(X) \leq 0$
which follows by the constraint at $Y=\emptyset$.
The extreme points of the submodular subdifferential admit interesting characterizations. We shall denote a subgradient at $X$ by $h_X \in \partial_f(X)$. The extreme points of $\partial_f(Y)$ may be computed via a greedy algorithm: Let $\sigma$ be a permutation of $V$ that assigns the elements in $X$ to the first $|X|$ positions ($i \leq |X|$ if and only if $\sigma(i) \in X$) and 
$S^{\sigma}_{|X|} = X$. This chain defines an extreme point
$h^{\sigma}_X$ of $\partial_f(X)$ with entries: 
$h^{\sigma}_X(\sigma(i)) = 
f(S^{\sigma}_i) - f(S^{\sigma}_{i-1})$. Note that for every subgradient $h_X \in \partial_f(X)$, we can define a modular lower bound $m_X(Y) = f(X) + h_X(Y) - h_X(X), \forall Y \subseteq V$, which satisfies $m_X(Y) \leq f(Y), \forall Y \subseteq V$. Moreover, we have that $m_X(X) = f(X)$, and hence the subdifferential exactly correspond to the set of tight modular lower bounds of a submodular function, at a given set $X$. If we choose $h_X$ to be an extreme subgradient, the modular lower bound becomes $m_X(Y) = h_X(Y)$, resulting in a normalized modular function\footnote{A set function $h$ is said to be normalized if $h(\emptyset) = 0)$.}.

The subdifferential is defined via an exponential number of inequalities. A key observation however is that many of these inequalities are redundant. Define three polyhedra: $\partial_f^1(X), \partial_f^2(X)$ and $\partial_f^3(X)$ where the first polyhedra is defined via inequalities of all subsets $Y \subseteq X$, the second via inequalities of all subsets $Y \supseteq X$ and the third comprising of all other inequalities. We immediately have that $\partial_f(X) = \partial_f^1(X) \cap \partial_f^2(X) \cap \partial_f^3(X)$. However, when $f$ is submodular, the inequalities in $\partial_f^3(X)$ are in fact redundant in characterizing $\partial_f(X)$. 

\begin{theorem} \cite{fujishige2005submodular} \label{subdiff-reduceineq} Given a submodular function $f$, $\partial_f(X) = \partial_f^1(X) \cap \partial_f^2(X)$. 
\end{theorem}

The subdifferential at the emptyset has a special relationship since
$\partial_f(\emptyset) = \mathcal P_f$.  Similarly
$\partial_f(V) = \mathcal P_{f^{\#}}$, where
$f^{\#}(X) = f(V) - f(V \backslash X)$ is the submodular dual of $f$.
Furthermore, since $f^{\#}$ is a supermodular function, it holds that
$\partial_f(V)$ is a \emph{supermodular polyhedron} (for a
supermodular function $g$, the supermodular polyhedron
is defined as
$\mathcal P_g = \{x: x(X) \geq g(X), \forall X \subseteq V\}$).

Finally define:
$\partial_f^{\symmdiff(1, 1)}(X) = \{x \in \mathbb{R}^V : \forall j \in X, f(j | X \backslash j) \leq x(j) \text{ and } \forall j \notin X, f(j | X) \geq x(j)\}$.
Notice that $\partial_f^{\symmdiff(1, 1)}(X) \supseteq \partial_f(X)$
since we are reducing the constraints of the subdifferential. In
particular $\partial_f^{\symmdiff(1, 1)}(X)$ just considers $n$
inequalities, by choosing the sets $Y$ in the definition of $\partial_f(X)$ such
that $|Y \symmdiff X| = 1$ (i.e., Hamming distance one away from $X$). This polyhedron
will be useful in characterizing local minimizers of a submodular
function (see Section~\ref{submodminopt}) and motivating analogous
constructs for local maxima (see, for example,
Proposition~\ref{prop:local_max_poly}).

\subsection{Subdifferentials and Optimality Conditions} \label{submodminopt}
Fujishige~\cite{fujishige1984theory} provides some interesting characterizations to the optimality conditions for unconstrained submodular minimization, which can be thought of as a discrete analog to the KKT conditions. The result shows that a set $A \subseteq V$ is a minimizer of $f: 2^V \rightarrow \mathbb{R}$ if and only if: $\textbf{0} \in \partial_f(A)$. This immediately provides necessary and sufficient conditions for optimality of $f$: A set $A$ minimizes a submodular function $f$ if and only if $f(A) \leq f(B)$ for all sets $B$ such that $B \subseteq A$ or $A \subseteq B$. Analogous characterizations have also been provided for constrained forms of submodular minimization, and interested readers may look at~\cite{fujishige1984theory}. Finally, we can provide a simple characterization on the local minimizers of a submodular function. A set $A \subseteq V$ is a local minimizer\footnote{A set $A$ is a local minimizer of a submodular function if $f(X) \geq f(A), \forall X: |X \backslash A| \leq 1, \text{ and } |A \backslash X| = 1$, that is all sets $X$ no more than hamming-distance one away from $A$.} of a submodular function if and only if $\textbf{0} \in \partial_f^{\symmdiff(1, 1)}(A)$. As was shown in \cite{rkiyersemiframework2013}, a local minimizer of a submodular function, in the unconstrained setting, can be found efficiently in $O(n^2)$ complexity.


\section{Submodularity and Concavity}
In this section, we investigate several polyhedral aspects of submodular functions relating them to concavity, thus complementing the results from Section~\ref{subconvex}. This provides a complete picture on the relationship between submodularity, convexity and concavity.

\subsection{The submodular upper polyhedron}
\label{uppersubpolysec}
A first step in characterizing the concave aspects of a submodular function is the submodular upper polyhedron. Intuitively this is the set of tight modular upper bounds of the function,
and we define it as follows:
\begin{align}
\mathcal P^f = \{x \in \mathbb{R}^n: x(S) \geq f(S), \forall S \subseteq V\}
\label{eqn:submodular_upper_polyhedron}
\end{align}
The above polyhedron can in fact be defined for any set function. In particular, when $f$ is supermodular, we get what is known as the \emph{supermodular polyhedron} \cite{fujishige2005submodular}. 
Presently, we are interested in the case when $f$ is submodular and hence we call this the 
\emph{submodular upper polyhedron}. 
Interestingly this has a very simple characterization.
\begin{theorem}\label{antisubpoly}
Given a submodular function $f$, 
\begin{align}
\mathcal P^f = \{x \in \mathbb{R}^n: x(j) \geq f(j)\}
\label{eqn:submodular_upper_polyhedron_n_inequalities}
\end{align}
\end{theorem}
\begin{proof}
Given $x \in \mathcal P^f$ and a set $S$, 
we have 
$x(S) = \sum_{i \in S} x(i) \geq \sum_{i \in S} f(i)$, since $\forall i, x(i) \geq f(i)$
by Eqn.~\eqref{eqn:submodular_upper_polyhedron}.
Hence $x(S) \geq \sum_{i \in S} f(i) \geq f(S)$. Thus, the irredundant inequalities 
are the singletons.
\end{proof}

The submodular upper polyhedron has a particularly simple characterization due to 
the submodularity of $f$. In other words, 
this polyhedron is not \emph{polyhedrally tight} in that many of the 
defining inequalities are redundunt. 
This polyhedron alone is not particularly interesting to define a concave extension. 

We end this subsection by investigating the submodular upper polyhedron membership problem. Owing to 
its simplicity, this problem is particularly simple which might seem surprising at first 
glance since $x \in \mathcal P^f$ is equivalent to,
Eqn.~\eqref{eqn:submodular_upper_polyhedron},
checking if $\max_{X \subseteq V} [f(X) - x(X)] \leq 0$. This involves maximization of a submodular function
which is NP hard. However, surprisingly this particular instance is easy!
\begin{lemma}
Given a submodular function $f$ and vector $x$, let $X$ be a set such that $f(X) - x(X) > 0$. Then there exists an $i \in X: f(i) - x(i) > 0$.
\end{lemma}
\begin{proof}
Observe that $f(X) - x(X) \leq \sum_{i \in X} [f(i) - x(i)]$. Since the l.h.s.\ is greater than 0, it implies that $\sum_{i \in X} [f(i) - x(i)] > 0$. Hence there should exist an $i \in X$ such that $f(i) - x(i) > 0$.\looseness-1
\end{proof}
An interesting corollary of the above, is that it is in fact easy to check if the maximizer of a submodular function is greater than equal to zero. Given a submodular function $f$, the problem is whether $\max_{X \subseteq V} f(X) \geq 0$. This can easily be checked without resorting to submodular function maximization.\looseness-1
\begin{corollary}
Given a submodular function $f$ with $f(\emptyset) = 0$, $\max_{X \subseteq V} f(X) > 0$ if and only if there exists an $i \in V$ such that $f(i) > 0$.
\end{corollary}
This fact is true only for a submodular function. For general set functions, even when $f(\emptyset) = 0$, it could potentially require an exponential search to determine if $\max_{X \subseteq V} f(X) > 0$.

\subsection{The Submodular Superdifferentials}
\label{sec:subm-superd}

Given a submodular function $f$ we can characterize its superdifferential as follows. We denote for a set $X$, the superdifferential with respect to $X$ as $\partial^f(X)$. 
\begin{equation*}\label{supdiff-def}
\partial^f(X) = \{x \in \mathbb{R}^n: f(Y) - x(Y) \leq f(X) - x(X), \forall Y \subseteq V \}
\end{equation*}
This characterization is analogous to the subdifferential of a submodular function (section~\ref{sec:subm-subd}). This is also akin to the superdifferential of a continuous concave function.  

Each supergradient $g_X \in \partial^f(X)$, defines a modular upper bound of a submodular function. In particular, define $m^X(Y) = g_X(Y) + f(X) - g_X(X)$. Then $m^X(Y)$ is a modular function which satisfies $m^X(Y) \geq f(Y), \forall Y \subseteq X$ and $m^X(X) = f(X)$. We note that $(x(v_1), x(v_2), \dots, x(v_n)) =
(f(v_1),f(v_2),\dots,f(v_n)) \in
\partial^f(\emptyset)$ which shows that at least $\partial^f(\emptyset)$ exists.
A bit further below (specifically Theorem~\ref{altviewsthm2}) we show that for any submodular function, $\partial^f(X)$ is non-empty
for all $X \subseteq V$.

Note that the superdifferential is defined by an exponential (i.e., $2^{|V|}$) number of inequalities. However owing to the submodularity of $f$ and akin to the subdifferential of $f$, we can reduce the number of inequalities. Define three polyhedrons as: $\partial^f_1(X) = \{x \in \mathbb{R}^n: f(Y) - x(Y) \leq f(X) - x(X), \forall Y \subseteq X \}$, $\partial^f_2(X) = \{x \in \mathbb{R}^n: f(Y) - x(Y) \leq f(X) - x(X), \forall Y \supseteq X \}$, and $\partial^f_3(X) = \{x \in \mathbb{R}^n: f(Y) - x(Y) \leq f(X) - x(X), \forall Y: Y \not \subseteq X, Y \not \supseteq X\}$. A trivial observation is that: $\partial^f(X) = \partial^f_1(X) \cap \partial^f_2(X) \cap \partial^f_3(X)$. As we show below for a submodular function $f$, $\partial^f_1(X)$ and $\partial^f_2(X)$ are actually very simple polyhedra.
\begin{theorem}\label{superdiff-reducedineq}
For a submodular function $f$, 
\begin{align}
\partial^f_1(X) &= \{x \in \mathbb{R}^n: f(j | X \backslash j) \geq x(j), \forall j \in X\} \\
\partial^f_2(X) &= \{x \in \mathbb{R}^n: f(j | X) \leq x(j), \forall j \notin X\}.
\end{align}
\end{theorem}
\begin{proof}
Consider $\partial^f_1(X)$. Notice that the inequalities defining the polyhedron can be rewritten as $\partial^f_1(X) = \{x \in \mathbb{R}^n:x(X \backslash Y) \leq f(X) - f(Y), \forall Y \subseteq X \}$. We then have that $x(X \backslash Y) = \sum_{j \in X \backslash Y} x(j) \leq \sum_{j \in X \backslash Y} f(j | X \backslash j)$, since $\forall j \in X, x(j) \leq f(j | X \backslash j)$ (this follows by considering only the subset of inequalities of $\partial^f_1(X)$ with sets $Y$ such that $|X \backslash Y| = 1$). Hence $x(X \backslash Y) \leq \sum_{j \in X \backslash Y} f(j | X \backslash j) \leq f(X) - f(Y)$. Hence 
an irredundant set of inequalities include those defined only through the singletons. 

 In order to show the characterization for $\partial^f_2(X)$, we have that $\partial^f_2(X) = \{x \in \mathbb{R}^n:x(Y \backslash X) \geq f(Y) - f(X), \forall Y \supseteq X \}$. It then follows that, $x(Y \backslash X) = \sum_{j \in Y \backslash X} x(j) \geq \sum_{j \in Y \backslash X} f(j | X)$, since $\forall j \notin X, x(j) \geq f(j | X)$. Hence $x(X \backslash Y) \geq \sum_{j \in Y \backslash X} f(j | X) \geq f(Y) - f(X)$, and again, an irredundant set of inequalities include those defined only through the singletons. 
\end{proof}

The above result significantly reduces the inequalities governing $\partial^f(X)$, and in fact, the polytopes $\partial^f_1(X)$ and $\partial^f_2(X)$ are very simple polyhedra. Recall that this is analogous to the submodular subdifferential (Theorem~\ref{subdiff-reduceineq}), where again owing to submodularity the number of inequalities are reduced significantly. In that case, we just need to consider the sets $Y$ which are subsets and supersets of $X$. It is interesting to note the contrast between the redunduncy of inequalities in the subdifferentials and the superdifferentials. In particular, here, the inequalities corresponding to sets $Y$ being the subsets and supersets of $X$ are mostly redundunt, while the non-redundunt ones are the rest of the inequalities. In other words, in the case of the subdifferential, $\partial_f^1(X)$ and $\partial_f^2(X)$ were non-redundunt, while $\partial_f^3(X)$ was entirely redundunt given the first two.
In the case of the superdifferentials, $\partial^f_1(X)$ and $\partial^f_2(X)$ are mostly internally redundant (they can be represented using only by $n$ inequalities), while $\partial^f_3(X)$ has no redundancy in general.

Unlike the subdifferentials, we cannot expect a closed form expression for the extreme points of $\partial^f(Y)$ (we provide several examples in the extended version of this paper). Moreover, they also seem to be hard to characterize algorithmically. For example, the superdifferential membership problem is NP hard.
\begin{lemma}\label{NPsuperdiffmembership}
Given a submodular function $f$ and a set $Y: \emptyset \subset Y \subset V$, the membership problem $y \in \partial^f(Y)$ is NP hard.
\end{lemma}
\begin{proof}
Notice that the membership problem $y \in \partial^f(Y)$ is equivalent to asking $\max_{X \subseteq V} f(X) - y(X) \leq f(Y) - y(Y)$. In other words, this is equivalent to asking if $Y$ is a maximizer of $f(X) - y(X)$ for a given vector $y$. This is the decision version of the submodular maximization problem and correspondingly is NP hard when $\emptyset \subset Y \subset V$. 
\end{proof}

Given that the membership problem is NP hard, it is also NP hard to solve a linear program over this polyhedron~\cite{grotschel1984geometric, schrijver2003combinatorial}.
The superdifferential of the empty and ground set, however, can be characterized easily:
\begin{lemma}
For any submodular function $f$ such that $f(\emptyset) = 0$, $\partial^f(\emptyset) = \{x \in \mathbb{R}^n: f(j) \leq x(j), \forall j \in V\}$. Similarly $\partial^f(V) = \{x \in \mathbb{R}^n: f(j | V \backslash j) \geq x(j), \forall j \in V\}$. 
\end{lemma}
This lemma is a direct consequence of Theorem~\ref{superdiff-reducedineq}.

While it is difficult to characterize the superdifferentials for sets $\emptyset \subset X \subset V$, we can provide inner and outer bounds of the superdifferentials.

\subsubsection{Outer Bounds on the Superdifferential}
\label{sec:outer-bounds-superd}

It is possible to provide a number of useful and practical outer bounds on the superdifferential. Recall that $\partial^f_1(Y)$ and $\partial^f_2(Y)$ are already simple polyhedra. We can then provide outer bounds on $\partial^f_3(Y)$ that,
together with $\partial^f_1(Y)$ and $\partial^f_2(Y)$,
provide simple bounds on $\partial^f(Y)$. Define for $1 \leq k,l \leq n$: $\partial^f_{3, \symmdiff(k, l)}(X) = \{x \in \mathbb{R}^n: f(Y) - x(Y) \leq f(X) - x(X), \forall Y: Y \not \subseteq X, Y \not \supseteq X, |Y \backslash X| \leq k-1, |X \backslash Y| \leq l-1\}$.

We can then define the outer bound: 
\begin{align}
\partial^f_{\symmdiff(k, l)}(X) = \partial^f_1(X) \cap \partial^f_2(X) \cap \partial^f_{3, \symmdiff(k, l)}(X). 
\end{align}
Observe that $\partial^f_{\symmdiff(k, l)}(X)$ is expressed in terms of $O(n^{k+l})$ inequalities, and hence for a given $k, l$ we can obtain the representation of $\partial^f_{\symmdiff(k, l)}(X)$ in polynomial time. We will see that this provides us with a heirarchy of outer bounds on the superdifferential:
\begin{theorem}
For a submodular function $f$ the following hold: a) $\partial^f_{\symmdiff(1, 1)}(X) = \partial^f_1(X) \cap \partial^f_2(X)$, b) $\forall 1 \leq k^{\prime} \leq k , 1 \leq l^{\prime} \leq l, \partial^f(X) \subseteq \partial^f_{\symmdiff(k, l)}(X) \subseteq \partial^f_{\symmdiff(k^{\prime}, l^{\prime})}(X) \subseteq \partial^f_{\symmdiff(1, 1)}(X)$ and c) $\partial^f_{\symmdiff(n, n)}(X) = \partial^f(X)$.
\end{theorem}
\begin{proof}
The proofs of items 1 and 3 follow directly from definitions. To see item 2, notice that the polyhedra $\partial^f_{\symmdiff(k, l)}$ become tighter as $k$ and $l$ increase finally approaching the superdifferential.
\end{proof}

\subsubsection{Inner Bounds on the Superdifferential}
\label{sec:inner-bounds-superd}

While it is hard to characterize the extreme points of the superdifferential, we can provide some specific supergradients.  
Define three vectors as follows:\looseness-1
\begin{align}
\ggrow_X(j) = 
\begin{cases}
f(j |X - j) & \text{ if }  j \in X\\
f(j) & \text { if } j \notin X \\
\end{cases} \\
\gshrink_X(j) = 
\begin{cases}
f(j |V - j) & \text{ if }  j \in X\\
f(j|X) & \text { if } j \notin X\\
\end{cases}\\
\gthree_X(j) = 
\begin{cases}
f(j |V - j) & \text{ if }  j \in X\\
f(j) & \text { if } j \notin X\\
\end{cases}
\end{align}

Then we have the following theorem:
\begin{theorem} \label{altviewsthm2} 
For a submodular function $f$, $\ggrow_X, \gshrink_X, \gthree_X \in \partial^f(X)$. Hence for every submodular function $f$ and set $X$, $\partial^f(X)$ is non-empty.
\end{theorem}
\begin{proof}
For submodular $f$, the following bounds are known to hold~\cite{nemhauser1978}:
\begin{align*} 
 f(Y) \leq f(X) - \sum_{j \in X \backslash Y } f(j| X \backslash j) + \sum_{j \in Y \backslash X} f(j| X \cap Y), \\
 f(Y) \leq f(X) - \sum_{j \in X \backslash Y } f(j| X \cup Y \backslash j) + \sum_{j \in Y \backslash X} f(j | X) 
\end{align*}
Using submodularity, we can loosen these bounds further to provide tight modular:
\looseness-1
\begin{align*}
 f(Y) \leq f(X) - \sum_{j \in X \backslash Y } f(j | X - \{j\}) + \sum_{j \in Y \backslash X} f(j | \emptyset)  \\
 f(Y) \leq f(X) - \sum_{j \in X \backslash Y } f(j | V- \{j\}) + \sum_{j \in Y \backslash X} f(j | X) \\
 f(Y) \leq f(X) - \sum_{j \in X \backslash Y } f(j | V- \{j\}) + \sum_{j \in Y \backslash X} f(j | \emptyset). 
\end{align*}
From the three bounds above, and substituting the expressions of the supergradients, we may
immediately verify that these are supergradients, namely that $\ggrow_X, \gshrink_X, \gthree_X \in \partial^f(X)$.
\end{proof}

Next, we provide inner bounds using the super-gradients defined above. Define two polyhedra:
\begin{align}
\partial^f_{\emptyset}(X) &= \{x \in \mathbb{R}^n: f(j) \leq x(j), \forall j \notin X\}, \\
\partial^f_V(X) &= \{x \in \mathbb{R}^n: f(j | V \backslash j) \geq x(j), \forall j \in X\}.
\end{align}
Then define: $\partial^f_{i, 1}(X) = \partial^f_1(X) \cap \partial^f_{V}(X)$, $\partial^f_{i, 2}(Y) = \partial^f_2(Y) \cap \partial^f_{\emptyset}(Y)$ and $\partial^f_{i, 3}(Y) = \partial^f_V(Y) \cap \partial^f_{\emptyset}(Y)$. Then note that $\partial^f_{i, 1}(Y)$ is a polyhedron with $\ggrow_Y$ as an extreme point. Similarly $\partial^f_{i, 2}(Y)$ has $\gshrink_Y$, while $\partial^f_{i, 3}(Y)$ has $\gthree_Y$ as its extreme points. All these are simple polyhedra, with a single extreme point. Also define: $\partial^f_{i, (1, 2)}(Y) = \text{conv}(\partial^f_{i, 1}(Y), \partial^f_{i, 2}(Y))$, where $\text{conv}(., .)$ represents the convex combination of two polyhedra\footnote{Given two polyhedra $\mathcal P_1, \mathcal P_2$, $\mathcal P = \text{conv}(\mathcal P_1, \mathcal P_2) = \{\lambda x_1 + (1 - \lambda) x_2, \lambda \in [0, 1], x_1 \in \mathcal P_1, x_2 \in \mathcal P_2\}$}. Then $\partial^f_{i, (1, 2)}(Y)$ is a polyhedron which has $\ggrow_Y$ and $\gshrink_Y$ as its extreme points. The following lemma characterizes the inner bounds of the superdifferential:
\begin{lemma}
Given a submodular function $f$, 
\begin{align}
\partial^f_{i, 3}(Y) \subseteq \partial^f_{i, 2}(Y) \subseteq \partial^f_{i, (1, 2)}(Y) \subseteq \partial^f(Y), \\
\partial^f_{i, 3}(Y) \subseteq \partial^f_{i, 1}(Y) \subseteq \partial^f_{i, (1, 2)}(Y) \subseteq \partial^f(Y)
\end{align}
\end{lemma}
\begin{proof}
The proof of this lemma follows directly from the definitions of the supergradients, corresponding polyhedra and submodularity.
\end{proof}
Finally we point out interesting connections between $\partial_f(X)$ and $\partial^f(X)$. Firstly, it is clear from the definitions that $\partial_f(X) \subseteq \partial_f^{\symmdiff(1, 1)}(X)$ and $\partial^f(X) \subseteq \partial^f_{\symmdiff(1, 1)}(X)$. Notice also that both $\partial_f^{\symmdiff(1, 1)}(X)$ and $\partial^f_{\symmdiff(1, 1)}(X)$ are simple polyhedra with a single extreme point,
\begin{align}
\tilde{g}_X(j) = 
\begin{cases}
f(j |X - j) & \text{ if }  j \in X\\
f(j | X) & \text { if } j \notin X\\
\end{cases}
\end{align}
The point $\tilde{g}_X$, is in general, neither a subgradient nor a supergradient at $X$. However both the semidifferentials are contained within (different) polyhedra defined via $\tilde{g}_X$. Please refer to the extended version for more details and figures demonstrating this.

While it is hard to characterize superdifferentials of general submodular functions, certain subclasses have some nice characterizations. An important such subclass if the class of $M^{\natural}$-concave functions~\cite{murota2003discrete}. These include a number of special cases like matroid rank functions, concave over cardinality functions etc. In some sense, these functions very closely resemble concave functions. The following result provides a compact representation of the superdifferential of these functions.
\begin{lemma}
Given a submodular function $f$ which is $M^{\natural}$-concave on $\{0,1\}^V$, its superdifferential satisfies,
\begin{align}
\partial^f(X) = \partial^f_{\symmdiff(2, 2)}(X)
\end{align}
In particular, it can be characterized via $O(n^2)$ inequalities.
\end{lemma}

\subsection{Optimality Conditions for submodular maximization}\label{submodmaxopt}
Just as the subdifferential of a submodular function provides optimality conditions for submodular miinimization, the superdifferential provides the optimality conditions for submodular maximization. 

We start with unconstrained submodular maximization:
\begin{equation}
\max_{X \subseteq V} f(X)
\end{equation}
Given a submodular function, we can give the KKT like conditions for submodular maximization: For a submodular function $f$, a set $A$ is a maximizer of $f$, if $\textbf{0} \in \partial^f(A)$.
However as expected, finding the set $A$, with the property above, or even verifying if for a given set $A$, $0 \in \partial^f(A)$ are both NP hard problems (from Lemma~\ref{NPsuperdiffmembership}). However thanks to the submodularity, we show that the outer bounds on the super-differential provide approximate optimality conditions for submodular maximization. Moreover, these bounds are easy to obtain.

\begin{proposition}
\label{prop:local_max_poly}
For a submodular function $f$, if $0 \in  \partial^f_{(1, 1)}(A)$ then $A$ is a local maxima of $f$ (that is, $\forall B \supseteq A, f(A) \geq f(B), \& \forall C \subseteq A, f(A) \geq f(C)$). Furthermore, if we define 
$S = \mbox{argmax}_{X \in \{ A, V \setminus A \}} f(A)$, then $f(S) \geq \frac{1}{3} OPT$ where $OPT$ is the optimal value. 
\end{proposition}
\begin{proof}
The local optimality condition follows directly from the definition of $\partial^f_{(1, 1)}(A)$ and the approximation guarantee follows from Theorem 3.4 in~\cite{janvondrak}. 
\end{proof}
The above result is interesting observation, since a very simple outer bound on the superdifferential, leads us to an approximate optimality condition for submodular maximization. We can also provide an interesting sufficient condition for the maximizers of a submodular function.
\begin{lemma}
If for any set $A$, $\textbf{0} \in \partial^f_{i,(1, 2)}(A)$, then $A$ is the global maxima of the submoduar function. In particular, if a local maxima $A$ is found (which is typically easy to do), it is guaranteed to be a global maxima, if it happens that $\textbf{0} \in \partial^f_{i,(1, 2)}(A)$. 
\end{lemma}
\begin{proof}
This proof follows from the fact that $\partial^f_{i,(1, 2)}(A) \subseteq \partial^f(A) \subseteq \partial^f_{(1, 1)}(A)$. 
Thus, if $\textbf{0} \in \partial^f_{i,(1, 2)}(A)$, it must also belong to $\partial^f(A)$, which means $A$ is the global optimizer of $f$.
\end{proof}
We can also provide similar results for constrained submodular maximization, which we omit in the interest of space. For more details see the longer version of this paper~\cite{iyer2015polyhedral}.

\bibliographystyle{./bibliography/IEEEtran}
\bibliography{./bibliography/IEEEabrv,./bibliography/IEEEexample}

\end{document}